\documentclass{article}
\usepackage{amsthm}
\usepackage{amsmath}
\usepackage{amsfonts}
\usepackage{amssymb}
\usepackage{algpseudocode}
\usepackage{algorithm}
\usepackage{fullpage}
\usepackage{xspace}
\usepackage{authblk}

\makeatletter
\def\algbackskip{\hskip-\ALG@thistlm}
\makeatother

\newtheorem{claim}{Claim}
\newtheorem{fact}{Fact}
\newtheorem{lemma}{Lemma}
\newtheorem{theorem}{Theorem}
\newtheorem{definition}{Definition}
\newtheorem{corollary}{Corollary}

\DeclareMathOperator*{\E}{\mathrm{E}}
\newcommand{\eps}{\varepsilon}

\newcommand{\bit}{\{0,1\}}
\newcommand{\bitsn}{\bit^n}

\newcommand{\event}{\mathcal{E}}

\newcommand{\deq}{:=}
\DeclareMathOperator*{\argmax}{argmax}

\DeclareMathOperator{\leaf}{\mathrm{leaf}}
\DeclareMathOperator{\adv}{adv}

\newcommand{\ygetix}{y^{(i \gets x)}\xspace}
\newcommand{\ygetsix}{y^{(i^* \gets x)}\xspace}
\newcommand{\zgetsix}{z^{(i^* \gets x)}\xspace}

\DeclareMathOperator{\xor}{\mathrm{XOR}}

\DeclareMathOperator{\xorg}{\xor\circ\xspace g}

\newcommand{\cA}{\mathcal{A}}

\newcommand{\PostBPP}{\mathrm{PostBPP}}
\newcommand{\noisyR}{\mathrm{noisyR}}
\newcommand{\fbs}{\mathrm{fbs}}
\newcommand{\ddt}{\mathrm{D}}
\newcommand{\rdt}{\mathrm{R}}
\newcommand{\rdta}{\overline{\mathrm{R}}}
\newcommand{\rdtde}{\rdt_{\delta,\eps}}

\title{A Strong XOR Lemma for Randomized Query Complexity}
\date{}
\author[1]{Joshua Brody}
\author[1]{Jae Tak Kim}
\author[1]{Peem Lerdputtipongporn}
\author[1]{Hariharan Srinivasulu}
\affil[1]{Swarthmore College}

\begin{document}
\maketitle

\begin{abstract}
  We give a strong direct sum theorem for computing $\xorg$.
  Specifically, we show that for every function $g$ and every
  $k\geq 2$, the randomized query complexity of computing the $\xor$
  of $k$ instances of $g$ satisfies
  $\rdta_\eps(\xorg) = \Theta(k \rdta_{\frac{\eps}{k}}(g))$.  This
  matches the naive success amplification upper bound and answers a conjecture
  of Blais and Brody~\cite{BlaisB19}.

  As a consequence of our strong direct sum theorem, we give a total
  function $g$ for which
  $\rdt(\xorg) = \Theta(k \log(k)\cdot \rdt(g))$, answering an open
  question from Ben-David et al.~\cite{BenDavidGKW20}.
 \end{abstract}

\section{Introduction}

We show that $\xor$ admits a strong direct sum theorem for randomized
query complexity.  Generally, the \emph{direct sum problem} asks how
the cost of computing a function $g$ scales with the number $k$ of
instances of the function that we need to compute.  This is a
foundational computational problem that has received considerable
attention~\cite{ImpagliazzoRW94,BenAsherN95,NisanRS99, Shaltiel03,
  JainKS10, BenDavidK18,Drucker12, BlaisB19, BenDavidB20a, BenDavidB20b,
  BenDavidGKW20}, including recent work of Blais and
Brody~\cite{BlaisB19}, which showed that \emph{average-case}
randomized query complexity obeys a direct sum theorem in a strong
sense --- computing $k$ copies of a function $g$ with overall error
$\eps$ requires $k$ times the cost of computing $g$ on one input with
very low ($\frac{\eps}{k}$) error.  This matches the naive success
amplification algorithm which runs an $\frac{\eps}{k}$-error algorithm for
$f$ once on each of $k$ inputs and applies a union bound to get an
overall error guarantee of $\eps$.

What happens if we don't need to compute $g$ on all instances, but
only on a \emph{function} $f\circ g$ of those instances?  Clearly the
same success amplification trick (compute $g$ on each input with low
error, then apply $f$ to the answers) works for computing $f\circ g$;
however, in principle, computing $f\circ g$ can be easier than
computing each instance of $g$ individually.  When a function
$f\circ g$ requires success amplification for all $g$, we say that $f$
\emph{admits a strong direct sum theorem}.  Our main result shows that
$\xor$ admits a strong direct sum theorem.

\subsubsection*{Query Complexity}A \emph{query algorithm} also known as
\emph{decision tree} computing $f$ is an algorithm $\cA$ that takes an
input $x$ to $f$, examines (or \emph{queries}) bits of $x$, and
outputs an answer for $f(x)$.  A \emph{leaf} of $\cA$ is a bit string
$q \in \bit^*$ representing the answers to the queries made by $\cA$
on input $x$.  Naturally, our general goal is to minimize the length
of $q$ i.e., minimize the number of queries needed to compute $f$.

A randomized algorithm $\cA$ \emph{computes} a function
$f : \bitsn \to \bit$ \emph{with error $\epsilon \ge 0$} if for every
input $x \in \bitsn$, the algorithm outputs the value $f(x)$ with
probability at least $1-\epsilon$.  The \emph{query cost} of $\cA$ is
the maximum number of bits of $x$ that it queries, with the
maximum taken over both the choice of input $x$ and the internal
randomness of $\cA$. The \emph{$\epsilon$-error (worst-case)
  randomized query complexity} of $f$ (also known as the
\emph{randomized decision tree complexity} of $f$) is the minimum
query complexity of an algorithm $\cA$ that computes $f$ with error at
most $\epsilon$. We denote this complexity by $\rdt_\epsilon(f)$, and
we write $\rdt(f) := \rdt_{\frac13}(f)$ to denote the $\frac13$-error
randomized query complexity of $f$.

Another natural measure for the query cost of a randomized algorithm
$\cA$ is the \emph{expected} number of coordinates of an input $x$
that it queries. Taking the maximum expected number of coordinates
queried by $\cA$ over all inputs yields the \emph{average query cost}
of $\cA$. The minimum average query complexity of an algorithm $\cA$
that computes a function $f$ with error at most $\epsilon$ is the
\emph{average $\epsilon$-error query complexity} of $f$, which we
denote by $\rdta_{\epsilon}(f)$. We again write
$\rdta(f) := \rdta_{\frac 13}(f)$. Note that
$\rdta_0(f)$ corresponds to the standard notion of
\emph{zero-error randomized query complexity} of $f$.

\subsection{Our Results}

Our main result is a strong direct sum theorem for $\xor$.

\begin{theorem}\label{thm:sds-xor}
  For every function $g:\bitsn \rightarrow \bit$ and all $\eps>0$, we
  have $\rdta_\eps(\xorg) = \Omega(k\cdot \rdta_{\eps/k}(g)).$
\end{theorem}

This answers Conjecture 1 of Blais and Brody~\cite{BlaisB19} in the
affirmative.

We prove Theorem~\ref{thm:sds-xor} by proving an analogous result in
distributional query complexity.  We also allow our algorithms to
\emph{abort} with constant probability.  Let
$\ddt_{\delta,\eps}^\mu(f)$ denote the minimal query cost of a
deterministic query algorithm that aborts with probability at most
$\delta$ and errs with probability at most $\eps$, where the
probability is taken over inputs $X \sim \mu$.  Similarly, let
$\rdt_{\delta,\eps}(f)$ denote the minimal query cost of a randomized
algorithm that computes $f$ with abort probability at most $\delta$
and error probability at most $\eps$ (here probabilities are taken
over the internal randomness of the algorithm).

Our main technical result is the following strong direct sum result
for $\xorg$ for distributional algorithms.
\begin{lemma}[Main Technical Lemma, informally stated.]\label{lemma:main}
  For every function $g:\bitsn \rightarrow \bit$, every distribution
  $\mu$, and every small enough $\delta,\eps>0$, we
  have
  $$\ddt_{\delta,\eps}^{\mu^k}(\xorg) =
  \Omega(k\ddt_{\delta',\eps'}^\mu(g))\ ,$$ for $\delta' = \Theta(1)$
  and $\eps' = \Theta(\eps/k)$.
\end{lemma}

In~\cite{BlaisB19}, Blais and Brody also gave a total function
$g:\bitsn\rightarrow \bit$ whose average $\eps$ error query complexity
satisfies $\rdta_\eps(g) = \Omega(\rdt(g)\cdot \log \frac{1}{\eps})$.
We use our strong XOR Lemma together with this function show the following.

\begin{corollary}\label{cor:application}
  There exists a total function $g:\bitsn \rightarrow \bit$ such that
  $\rdt_\eps(\xorg) = \Omega(k\log(k)\cdot \rdt_\eps(g))$.
\end{corollary}
\begin{proof}
  Let $g: \bitsn \rightarrow \bit$ be a function guaranteed
  by~\cite{BlaisB19}.  Then, we have
  $$\rdt(\xorg) \geq \rdta(\xorg) \geq \Omega(k\cdot \rdta_{1/3k}(g)) \geq
  \Omega(k\cdot \rdt(g)\cdot \log(3k)) = \Omega(k\log(k)\cdot \rdt(g))\ ,$$
  where the second inequality is by Theorem~\ref{thm:sds-xor} and the
  third inequality is from the query complexity guarantee of $g$.
\end{proof}

This answers Open Question 1 from recent work of Ben-David et
al.~\cite{BenDavidGKW20}.

\subsection{Previous and Related Work}

Jain et al.~\cite{JainKS10} gave direct sum theorems for deterministic
and randomized query complexity.  While their direct sum result holds
for worst-case randomized query complexity, they incur an
\emph{increase} in error
($\rdt_\eps(f^k) \geq \delta \cdot k \cdot \rdt_{\eps + \delta}(f)$)
when computing a single copy of $f$.  Shaltiel~\cite{Shaltiel03} gave
a counterexample function for which direct sum fails to hold for
distributional complexity.  Drucker~\cite{Drucker12} gave a strong
\emph{direct product} theorem for randomized query complexity.

Our work is most closely related to that of Blais and
Brody~\cite{BlaisB19}, who give a strong direct sum theorem for
$\rdta_\eps(f^k) = \Omega(k\rdta_{\eps/k}(f))$, and explicitly
conjecture that $\xor$ admits a strong direct product theorem.
Both~\cite{BlaisB19} and ours use techniques similar to work of
Molinaro et al.~\cite{MolinaroWY13, MolinaroWY15} who give strong
direct sum theorems for communication complexity.

Our strong direct sum for $\xor$ is an example of a \emph{composition
  theorem}---lower bound on the query complexity of functions of the
form $f\circ g$.  Several very recent works studied composition
theorems in query complexity.  Bassilakis et
al.~\cite{BassilakisDGHMT20} show that
$\rdt(f\circ g) = \Omega(\fbs(f)\rdt(g))$, where $\fbs(f)$ is the
\emph{fractional block sensitivity} of $f$.  Ben-David and
Blais~\cite{BenDavidB20a, BenDavidB20b} give a tight lower bound on
$\rdt(f\circ g)$ as a product of $\rdt(g)$ and a new measure they
define called \emph{$\noisyR(f)$}, which measures the complexity of
computing $f$ on noisy inputs.  They also characterize $\noisyR(f)$ in
terms of the gap-majority function.  Ben-David et
al~\cite{BenDavidGKW20} explicitly consider strong direct sum theorems
for composed functions in randomized query complexity, asking whether
the naive success amplification algorithm is necessary to compute
$f\circ g$.  They give a partial strong direct sum theorem, showing
that there exists a partial function $g$ such that computing $\xorg$
requires success amplification, even in a model where the abort
probability may be arbitrarily close to 1.\footnote{In this query
  complexity model, called $\PostBPP$, the query algorithm is allowed
  to abort with any probability strictly less than 1.  When it doesn't
  abort, it must output $f$ with probability at least $1-\eps$.}
Ben-David et al. explicitly ask whether there exists a total function
$g$ such that $\rdt(\xorg) = \Omega(k\log(k)\rdt(g))$.

\subsection{Our Technique.}  Our technique most closely follows
the strong direct sum theorem of Blais and Brody.  We start with a
query algorithm that computes $\xorg$ and use it to build a query
algorithm for computing $g$ with low error.  To do this, we'll take an
input for $g$ and \emph{embed} it into an input for $\xorg$.  Given
$x \in \bitsn$, $i \in [k]$, and $y \in \bit^{n\times k}$, let
$\ygetix \deq (y^{(1)}, \ldots, y^{(i-1)}, x, y^{(i+1)}, \ldots
y^{(k)})$ denote the input obtained from $y$ by replacing the $i$-th
coordinate $y^{(i)}$ with $x$.  Note that if $x \sim \mu$ and
$y \sim \mu^k$,\footnote{We use $\mu^k$ to denote the distribution on
  $k$-tuples where each coordinate is independently distributed
  $\sim \mu$.} then $\ygetix\sim \mu^k$ for all $i \in [k]$.

We require the following observation of Drucker~\cite{Drucker12}.
\begin{lemma}[\cite{Drucker12}, Lemma 3.2]
  Let $y\sim \mu^k$ be an input for a query algorithm $\cA$, and
  consider any execution of queries by $\cA$.  The distribution of
  coordinates of $y$, conditioned on the queries made by $\cA$,
  remains a product distribution.
\end{lemma}

In particular, the answers to $g(y^{(i)})$ remain independent bits
conditioned on any set of queries made by the query algorithm.
Our first observation is that in order to compute $\xorg(y)$ with high
probability, we must be able to compute $g(y^{(i)})$ with very high
probability for many $i$'s.  The intuition behind this observation is
captured by the following simple fact about the $\xor$ of independent
random bits.

Define the \emph{bias} of a random bit $X \in \bit$ as $r(X) \deq
\max_{b\in\bit}\Pr[X=b]$.  Define the \emph{advantage} of $X$ as $\adv(X)
\deq 2r(X)-1$.  Note that when $\adv(X) = \delta$, then $r(X) = \frac{1}{2}(1+\delta)$.

\begin{fact}\label{fact:xor-adv}
  Let $X_1,\ldots, X_k$ bit independent random bits, and let $a_i$ be
  the advantage of $X_i$.  Then, $$\adv(X_1\oplus \cdots \oplus X_k) =
  \prod_{i=1}^k \adv(X_i)\ .$$
\end{fact}

For completeness, we provide a proof of Fact~\ref{fact:xor-adv} in
Appendix~\ref{sec:facts}.

Given an algorithm for $\xorg$ that has error $\eps$, it follows that
for typical leaves the advantage of computing $\xorg$ is
$\gtrsim 1-2\eps$.  Fact~\ref{fact:xor-adv} shows that for such
leaves, the advantage of computing $g(y^{(i)})$ for most coordinates
$i$ is $\gtrsim (1-2\eps)^{1/k} = 1-\Theta(\eps/k)$.  Thus,
conditioned on reaching this leaf of the query algorithm, we could
compute $g(y^{(i)})$ with very high probability.  We'd like to fix a
coordinate $i^*$ such that for most leaves, our advantage in computing
$g$ on coordinate $i^*$ is $1-O(\eps/k)$.  There are other
complications, namely that (i) our construction needs to handle aborts
gracefully and (ii) our construction must ensure that the algorithm
for $\xorg$ doesn't query the $i^*$-th coordinate too many times.  Our
construction identifies a coordinate $i^*$ and a string
$z \in \bit^{n\times k}$, and on input $x \in \bitsn$ it emulates a
query algorithm for $\xorg$ on input $\zgetsix$, and outputs our best
guess for $g(x)$ (which is now $g$ evaluated on coordinate $i^*$ of
$\zgetsix$), aborting when needed e.g., when the algorithm for
$\xorg$ aborts or when it queries too many bits of $x$.  We defer
full details of the proof to Section~\ref{sec:main}.

\subsection{Preliminaries and Notation}

Suppose that $f$ is a Boolean function on domain $\bitsn$ and that
$\mu$ is a distribution on $\bitsn$.  Let $\mu^k$ denote the
distribution obtained on $k$-tuples of $\bitsn$ obtained by sampling
each coordinate independently according to $\mu$.

An algorithm $\cA$ is a $[q,\delta,\eps,\mu]$-distributional query
algorithm for $f$ if $\cA$ is a deterministic algorithm with query
cost $q$ that computes $f$ with error probability at most $\eps$ and
abort probability at most $\delta$ when the input $x$ is drawn from
$\mu$.  We write $\cA(x) = \bot$ to denote that $\cA$ aborts on input
$x$.

Our main theorem is a direct sum result for $\xorg$ for average case
randomized query complexity; however, Lemma~\ref{lemma:main} uses
distributional query complexity.  The following results from Blais and
Brody~\cite{BlaisB19} connect the query complexities in the
randomized, average-case randomized, and distributional query models.

\begin{fact}[\cite{BlaisB19}, Proposition 14]
\label{fact:abort-average}
For every function $f : \bitsn \to \bit$, every
$0 \le \epsilon < \frac12$ and every $0 < \delta < 1$,
\[
\delta \cdot \rdtde(f) \le \overline{\rdt}_\epsilon(f) \le \tfrac{1}{1-\delta} \cdot \rdt_{\delta,(1-\delta)\epsilon}(f).
\]
\end{fact}

\begin{fact}[\cite{BlaisB19}, Lemma 15]
\label{fact:yao-abort}
For any $\alpha,\beta>0$ such that $\alpha+\beta \leq 1$, we have
\[
\max_\mu \ddt_{\delta/\alpha, \eps/\beta}^\mu(f) \leq \rdtde(f)
    \leq \max_{\mu} \ddt_{\alpha\delta, \beta\eps}^\mu(f).
\]
\end{fact}

We'll also use the following convenient facts about probability and
expectation.  For completeness we provide proofs in
Appendix~\ref{sec:facts}.
\begin{fact}\label{fact:prob-ex}
  Let $S,T$ be random variables.  Let
  $\event = \event(S,T)$ and $\cA$ be events, and for any $s$, let
  $\mu_s$ be the distribution on $T$ conditioned on $S=s$.  Then,
  $$ \Pr_{S,T}[\event|\cA] = \E_S\left[\Pr_{T\sim \mu_S}[\event(S,T) | \cA]\right]\ .$$
\end{fact}

\begin{fact}[Markov Inequality for Bounded Variables]\label{fact:bounded-markov}
  Let $X$ be a real-valued random variable with $0\leq X \leq 1$.
  Suppose that $E[X] \geq 1-\eps$.  Then, for any $T>1$ it holds that
  $$\Pr[X<1-T\eps] < \frac{1}{T}\ .$$
\end{fact}

\section{Strong XOR Lemma}\label{sec:main}

In this section, we prove our main result.

\begin{lemma}[Formal Restatement of Lemma~\ref{lemma:main}]
  For every function $g: \bitsn \rightarrow \bit$, every distribution
  $\mu$ on $\bitsn$, every $0 \leq \delta \leq \frac{1}{5}$, and every
  $0 < \eps \leq \frac{1}{800}$, we
  have
  $$\ddt_{\delta,\eps}^{\mu^k}(\xorg) = \Omega\left(k \cdot
    \ddt_{\delta',\eps'}^\mu(g)\right)\ ,$$ for
  $\delta' = 0.34+4\delta$ and $\eps' = \frac{320000\eps}{k}$.
\end{lemma}
\begin{proof}
  Let $q \deq D_{\delta,\eps}^{\mu^k}(\xorg)$, and suppose that $\cA$
  is a $[q,\delta,\eps,\mu^k]$-distributional query algorithm for
  $\xorg$.  Our goal is to construct an
  $\left[O(q/k), \delta', \eps', \mu\right]$-distributional query
  algorithm $\cA'$ for $g$.  Towards that end, for each leaf $\ell$ of
  $\cA$ define
  \begin{align*}
    b_\ell &\deq \argmax_{b \in \bit} \Pr_{x \sim \mu^k}[\xorg(x) = b | \leaf(\cA, x) = \ell] \\
    r_\ell &\deq \Pr_{x \sim \mu^k}[\xorg(x) = b_\ell | \leaf(\cA,x) = \ell] \\
    a_\ell &\deq 2r_\ell - 1\ .
  \end{align*}
  Call $a_\ell$ the \emph{advantage} of $\cA$ on leaf $\ell$.

  The purpose of $\cA$ is to compute $\xorg$; however, we'll show that
  $\cA$ must additionally be able to compute $g$ reasonably
  well on many coordinates of $x$.  For any $i \in [k]$ and any leaf $\ell$, define
  \begin{align*}
    b_{i,\ell} &\deq \argmax_{b \in \bit} \Pr_{x \sim \mu^k}[b = g(x^{(i)}) | \leaf(\cA,x) = \ell] \\
    r_{i,\ell} &\deq \Pr_{x\sim \mu^k}[b_{i,\ell} = g(x^{(i)}) | \leaf(\cA,x) = \ell] \\
    a_{i,\ell} &\deq 2r_{i,\ell} - 1\ .
  \end{align*}

  If $\cA$ reaches leaf $\ell$ on input $y$, then write $\cA(y)_i \deq
  b_{i,\ell}$.  $\cA(y)_i$ represents $\cA$'s best guess for $g(y^{(i)})$.

  Next, we define some structural characteristics of leaves that we'll
  need to complete the proof.

  \begin{definition}[Good leaves, good coordinates]
    \mbox{ }

    \begin{itemize}
    \item Call a leaf $\ell$ \emph{good} if $r_\ell \geq 1-200\eps$.
    \item Call a leaf $\ell$ \emph{good for $i$} if $a_{i,\ell} \geq 1-80000\eps/k$.
    \item Call coordinate $i$ \emph{good} if
      $\Pr_{x \sim \mu^k}[\leaf(\cA,x) \text{ is good for } i| \cA(x)
      \text{ doesn't abort}] \geq 1-\frac{3}{50}$.
    \end{itemize}
  \end{definition}
  When a leaf is good for $i$, then $\cA$, conditioned on reaching
  this leaf, computes $g(x^{(i)})$ with very high probability.  When a
  coordinate $i$ is good, then with high probability $\cA$ reaches a
  leaf that is good for $i$.  To make our embedding work, we need to
  fix a good coordinate $i^*$ such that $\cA$ makes only $O(q/k)$
  queries on this coordinate.  The following claim shows that most
  coordinates are good.

  \begin{claim}\label{claim:goodi}
    $i$ is good for at least $\frac{2}{3}k$ indices $i \in [k]$.
  \end{claim}

  We defer the proof of Claim~\ref{claim:goodi} to the following subsection.
  Next, for each $i \in [k]$, let $q_i(x)$ denote the number of queries that
  $\cA$ makes to $x^{(i)}$ on input $x$.  The query cost of $\cA$
  guarantees that for each input $x$, $\sum_{1\leq i\leq k} q_i(x)
  \leq q$.  Therefore, $\sum_{i\in [k]} E_{x\sim \mu^k}[q_i(x)] \leq
  q$, and so at least $\frac{2}{3}k$ indices $i \in [k]$
  satisfy
  \begin{equation}
    \label{eqn:querygoodi}
    \E_{x\sim\mu^k}[q_i(x)] \leq \frac{3q}{k}\ .
  \end{equation}
  Thus, there exists $i^*$ which satisfies both
  Claim~\ref{claim:goodi} and inequality~\eqref{eqn:querygoodi}.  Fix
  such an $i^*$.  For inputs $y\in \bit^{n\times k}$ and $x \in
  \bitsn$, let $\ygetsix \deq (y^{(1)}, \ldots, y^{(i^*-1)}, x,
  y^{(i^*+1)}, \ldots y^{(k)})$ denote the input obtained from $y$ by
  replacing $y^{(i^*)}$ with $x$. Note that if $y\sim \mu^k$ and $x
  \sim \mu$, then $\ygetix\sim \mu^k$ for all $i \in [k]$. With this
  notation and using Fact~\ref{fact:prob-ex}, the conditions from
  inequality~\eqref{eqn:querygoodi} and Claim~\ref{claim:goodi}
  satisfied by $i^*$ can be rewritten as

  $$\E_{y\sim \mu^k}\left[\E_{x\sim\mu}\left[q_{i^*}(\ygetsix)\right]\right] \leq \frac{3q}{k}\ ,$$ and
  $$\E_{y\sim \mu^k}\left[\Pr_{x \sim \mu}\left[\leaf\left(\cA,\ygetsix\right) \text{ is bad for } i^* | \cA(\ygetsix) \text{ doesn't abort}\right]\right] \leq \frac{3}{50}\ .$$

  Since $\cA$ has at most $\delta$ abort probability, we have
  $$\E_{y\sim \mu^k}\left[\Pr_{x \sim \mu}\left[\cA(\ygetsix) = \bot\right]\right] \leq \delta\ .$$

  Finally, for any leaf $\ell$ for which $i^*$ is good, we have
  $a_{i^*,\ell} \geq 1-80000\eps/k$.  Hence
  $$\E_{y\sim \mu^k}\left[\Pr_{x \sim \mu}\left[\cA(\ygetsix)_{i^*} \neq g(x) | \leaf\left(\cA, \ygetsix\right) \text{ is good for } i^*\right]\right] \leq \frac{80000\eps}{k}\ .$$

  Therefore by Markov's Inequality, there exists $z \in \bit^{n\times k}$ such that
  \begin{align}
    \E_{x \sim \mu}\left[q_{i^*}(\zgetsix)\right] &\leq \frac{12q}{k}\ ,\label{eq:query}\\
    \Pr_{x\sim \mu}\left[\leaf(\cA, \zgetsix) \text{ is bad for } i^* | \cA(\zgetsix) \neq \bot \right] &\leq \frac{6}{25}\ ,\label{eq:bad}\\
    \Pr_{x\sim \mu}\left[\cA(\zgetsix) = \bot\right] &\leq 4\delta\ , \text{ and }\label{eq:abort}\\
    \Pr_{x\sim \mu}\left[\cA(\zgetsix)_{i^*} \neq g(x) | \leaf(\cA,\zgetsix) \text{ is good for } i^*\right] &\leq \frac{320000\eps}{k}\ .\label{eq:correct}
  \end{align}
  Fix this $z$.  Now that $i^*$ and $z$ are fixed, we are ready to describe our algorithm.

\begin{algorithm}
\caption{$\cA'_{z,i^*}(x)$}
\begin{algorithmic}[1]
\State $y \gets \zgetsix$ \State Emulate algorithm $\cA$ on input $y$.
\State Abort if $\cA$ aborts, if $\cA$ queries more than $\frac{120q}{k}$ bits of $x$,
or if $\cA$ reaches a bad leaf.
\State Otherwise, output $\cA(y)$.
\end{algorithmic}
\end{algorithm}
Note that the emulation is possible since whenever $\cA$ queries the
$j$-th bit of $y^{(i^*)}$, we can query $x_j$, and we can emulate
$\cA$ querying a bit of $y^{(i)}$ for $i\neq i^*$ directly since $z$
is fixed.  It remains to show that $\cA'$ is a $\left[\frac{120q}{k},
  0.34+4\delta, \frac{320000\eps}{k}, \mu\right]$-distributional query
algorithm for $f$.

First, note that $\cA'$ makes at most $120q/k$ queries, since it
aborts instead of making more queries.  Next, consider the abort
probability of $\cA'$.  Our algorithm aborts if $\cA$ aborts, if $\cA$
probes more than $\frac{120q}{k}$ bits, or if $\cA$ reaches a bad
leaf.  By inequality~\eqref{eq:abort}, $\cA$ aborts with probability at most
$4\delta$.  By inequality~\eqref{eq:query} and Markov's Inequality, the
probability that $\cA$ probes $120q/k$ bits is at most $1/10$.
By inequality~\eqref{eq:bad}, we have $\Pr_{x \sim \mu}[\cA \text{ reaches a bad
    leaf}] \leq 6/25$.  Hence, $\cA'$ aborts with probability at most
$4\delta + \frac{1}{10} + \frac{6}{25} = 0.34 + 4\delta$.  Finally,
note that if $\cA'$ doesn't abort, then $\cA$ reaches a leaf which is
good for $i^*$.  By inequality~\eqref{eq:correct}, $\cA'$ errs with
probability at most $320000\eps/k$ in this case.

We have constructed an algorithm $\cA'$ for $g$ that makes at most
$120q/k$ queries, and when the input $x \sim \mu$, $\cA'$ aborts with
probability at most $\delta'$ and errs with probability at most
$\eps'$.  Hence, $D_{\delta',\eps'}^\mu(g) \leq 120q/k$.  Rearranging
terms and recalling that $q = D_{\delta,\eps}^{\mu^k}(\xorg)$, we get
$$D_{\delta,\eps}^{\mu^k}(\xorg) \geq \frac{k}{120}D_{\delta',\eps'}^\mu(g)\ ,$$ completing the proof.
\end{proof}

\subsection{Proof of Claim~\ref{claim:goodi}.}
\begin{proof}[Proof of Claim~\ref{claim:goodi}]
  Let $I$ be uniform on $[k]$.  We want to show that $\Pr[I \text{ is good}] \geq 2/3$.

  Conditioned on $\cA$ not aborting, it outputs the correct value of
  $\xorg$ with probability at least $1-\frac{\eps}{1-\delta} \geq
  1-2\eps$.  We first analyze this error probability by
  conditioning on which leaf is reached.  Let $\nu$ be the
  distribution on $\leaf(\cA,x)$ when $x \sim \mu^k$, conditioned on
  $\cA$ not aborting.  Let $L\sim \nu$.  Then, we have
\begin{align*}
  1-2\eps &\leq \Pr_{x\sim \mu^k}[\cA(x) = \xorg(x) | \cA \text{ doesn't abort}] \\
  &= \sum_{\text{leaf }\ell } \Pr_{L \sim \nu}[L = \ell]\cdot \Pr[\cA(x) = \xorg(x) | L=\ell] \\
  &= \sum_\ell \Pr[L=\ell]\cdot r_\ell \\
  &= \E_L[r_L]\ .
\end{align*}

Thus, $\E[r_L] \geq 1-2\eps$.  Recalling that $\ell$ is good if
$r_\ell \geq 1-200\eps$ and using Fact~\ref{fact:bounded-markov}, $L$
is good with probability at least $0.99$.  Note also that when $\ell$
is good, then $a_\ell \geq 1-400\eps$.  Let $\beta_\ell := \Pr_I[\ell \text{ is bad for } I]$.
Using $1+x\leq e^x$ and $e^{-2x} \leq 1-x$ (which holds for all $0\leq x \leq 1/2$), we have for any good leaf $\ell$
$$1-400\eps \leq a_\ell = \prod_{i=1}^k a_{i,\ell} \\ \leq
\left(1-\frac{80000\eps}{k}\right)^{k\beta_\ell} \\ \leq
e^{-80000\eps\cdot \beta_\ell} \\
\leq 1-40000\eps\beta_\ell\ .$$ Rearranging terms, we see that
$\beta_\ell \leq 0.01$.  We've just shown that a random leaf $\ell$ is
good with high probability, and when $\ell$ is good, it is good for
many $i$.  We need to show that there are many $i$ such that most
leaves are good for $i$.  Towards that end, let $\delta_{i,\ell} \deq
1$ if $\ell$ is good for $i$; otherwise, set $\delta_{i,\ell} \deq 0$.
\begin{align*}
  \E_I\left[\Pr_{x\sim\mu^k}[\leaf(\cA,x) \text{ good for } I| \cA \text{ doesn't abort}]\right] &= \E_I\left[\sum_{\ell} \Pr[L=\ell]\cdot \delta_{I,\ell}\right] \\
  &= \sum_{\ell}\Pr[L=\ell] \E_I[\delta_{I,\ell}] \\
  &= \sum_{\ell}\Pr[L=\ell]\Pr_I[\ell \text{ good for } I] \\
  &\geq \sum_{\text{good} \ell} \Pr[L=\ell]\cdot (1-\beta_\ell) \\
  &= \Pr_{L}[L \text{ is good}]\cdot (1-\beta_\ell) \\
  &\geq 0.99(1-\beta_\ell) \\
  &> 0.98\ .
\end{align*}
Thus, $\E_I\left[\Pr_{x\sim\mu^k}[\leaf(\cA,x) \text{ good for } I|
    \cA \text{ doesn't abort}]\right] \geq 1-\frac{1}{50}$.  Recalling
that $i$ is good if \\$\Pr[\leaf(\cA,x) \text{ good for } i | \cA(x)
  \text{ doesn't abort}] \geq 1-\frac{3}{50}$ and using
Fact~\ref{fact:bounded-markov}, it follows that $\Pr_I[I \text{ is
    good}] \geq 2/3$.  This completes the proof.
\end{proof}

\subsection{Proof of Theorem~\ref{thm:sds-xor}}
\begin{proof}[Proof of Theorem~\ref{thm:sds-xor}]
  Define $\eps' \deq 640000\eps$.  Let $\mu$ be the input distribution
  for $g$ achieving
  $\max_\mu D_{\frac{1}{2}, \frac{\eps'}{k}}^\mu(g)$, and let $\mu^k$
  be the $k$-fold product distribution of $\mu$.
  By the first inequality of Fact~\ref{fact:abort-average} and the
  first inequality of Fact~\ref{fact:yao-abort}, we
  have
  $$\rdta_\eps(\xorg) \geq \frac{1}{50}\rdt_{\frac{1}{50},\eps}(\xorg) \geq \frac{1}{50}\ddt_{\frac{1}{25}, 2\eps}^{\mu^k}(\xorg)\ .$$
  Additionally, by Lemma~\ref{lemma:main} and the second inequalities of Facts~\ref{fact:abort-average} and~\ref{fact:yao-abort}, we have
  $$\ddt_{\frac{1}{25}, 2\eps}^{\mu^k}(\xorg) \geq \frac{k}{120} \ddt_{\frac{1}{2},\frac{\eps'}{k}}^\mu(g) \geq \frac{k}{120}\rdt_{\frac{2}{3},\frac{4\eps'}{k}}(g) \geq \frac{k}{360}\rdta_{\frac{12\eps'}{k}}(g) \ .$$
  Thus, we have
  $\rdta_\eps(\xorg) =
  \Omega\left(\ddt_{\frac{1}{25},2\eps}^{\mu^k}(\xorg)\right)$ and
  $\ddt_{\frac{1}{25},2\eps}^{\mu^k}(\xorg) =
  \Omega\left(k\rdta_{\frac{12\eps'}{k}}(g)\right)$.  By standard
  success amplification
  $\rdta_{\frac{12\eps'}{k}}(g) = \Theta(\rdta_{\frac{\eps}{k}}(g))$.
  Putting these together yields
  $$\rdta_\eps(\xorg) =
  \Omega\left(\ddt_{\frac{1}{25},2\eps}^{\mu^k}(\xorg)\right) =
  \Omega\left(k\rdta_{\frac{12\eps'}{k}}(g)\right) =
  \Omega\left(\rdta_{\frac{\eps}{k}}(g)\right)\ ,$$ hence
  $\rdta_\eps(\xorg) = \Omega\left(k\rdta_{\frac{\eps}{k}}(g)\right)$ completing the proof.
\end{proof}

\section*{Acknowledgments}
The authors thank Runze Wang for several helpful discussions.

\bibliographystyle{plain}
\bibliography{xor}

\appendix

\section{Proofs of Technical Lemmas}\label{sec:facts}

\begin{proof}[Proof of Fact~\ref{fact:xor-adv}]
  For each $i$, let $b_i \deq \argmax_{b\in \bit} \Pr[X_i = b]$ and
  $\delta_i \deq \adv(X_i)$.  Then $\Pr[X_i = b_i] =
  \frac{1}{2}(1+\delta_i)$.
  We prove Fact~\ref{fact:xor-adv} by induction on $k$.  When $k=1$,
  there is nothing to prove.  For $k=2$, note that

  \begin{align*}
    \Pr[X_1 \oplus X_2 = b_1\oplus b_2] &=
    \frac{1}{2}(1+\delta_1)\frac{1}{2}(1+\delta_2) +
    \frac{1}{2}(1-\delta_1)\frac{1}{2}(1-\delta_2) \\
    &= \frac{1}{4}(1+\delta_1+\delta_2 + \delta_1\delta_2) + \frac{1}{4}(1-\delta_1-\delta_2 + \delta_1\delta_2) \\
    &= \frac{1}{2}(1+\delta_1\delta_2)\ .
  \end{align*}
  Hence $X_1\oplus X_2$ has advantage $\delta_1\delta_2$ and the claim
  holds for $k=2$.
  For an induction hypothesis, suppose that the claim holds for
  $X_1\oplus \cdots \oplus X_{k-1}$.  Then, setting $Y \deq X_1 \oplus
  \cdots \oplus X_{k-1}$, by the induction hypothesis, we have $\adv(Y)
  = \prod_{i=1}^{k-1} \adv(X_i)$.  Moreover, $X_1\oplus \cdots \oplus
  X_k = Y\oplus X_k$, and $$\adv(X_1 \oplus \cdots \oplus X_k) =
  \adv(Y\oplus X_k) = \adv(Y)\adv(X_k) = \prod_{i=1}^k \adv(X_i)\ .$$
\end{proof}

\begin{proof}[Proof of Fact~\ref{fact:prob-ex}]
  We condition $\Pr_{S,T}[\event(S,T)|\cA]$ on $S$.
  \begin{align*}
    \Pr_{S,T}[\event|\cA] &= \sum_s \Pr[S=s|\cA]\Pr_T[\event(S,T) | \cA,S=s] \\
                          &= \sum_s \Pr[S=s|\cA]\Pr_{T\sim \mu_s}[\event(S,T) | \cA] \\
                          &=\E_S\left[\Pr_{T \sim \mu_S}[\event(S,T)|\cA]\right]\ .
  \end{align*}

\end{proof}

\begin{proof}[Proof of Fact~\ref{fact:bounded-markov}]
  Let $Y \deq 1-X$.  Then, $E[Y] \leq \eps$.  By Markov's Inequality
  we have $$\Pr[X<1-T\eps] = \Pr[Y > T\eps] \leq \frac{1}{T}\ .$$
\end{proof}

\end{document}